\documentclass[12pt]{article}
\usepackage{amsmath}
\usepackage{graphicx,psfrag,epsf}
\usepackage{enumerate}
\usepackage{psfrag,epsf,amsmath}
\usepackage{url} 
\usepackage{amsmath,amssymb,amsfonts,amsthm,mathtools,mathrsfs}

\usepackage{booktabs}  
\usepackage{colortbl}  
\usepackage[table,xcdraw]{xcolor}  

\usepackage{bm,bbm}
\usepackage{color}
\usepackage{subfigure}
\usepackage{algorithm,algpseudocode}
\usepackage[symbol]{footmisc}
\usepackage{scalefnt}
\usepackage{authblk} 
\usepackage{multirow,centernot}
\usepackage[colorlinks=true, citecolor=blue, urlcolor=blue]{hyperref}
\usepackage{natbib}
\usepackage{dsfont}
\usepackage[title]{appendix}
\input cyracc.def

\usepackage[left=1in,top=1.1in,right=0.5in,bottom=1in]{geometry}

\theoremstyle{definition}

\newtheorem{theorem}{Theorem}[section]

\newtheorem{corollary}[theorem]{Corollary}

\newtheorem{proposition}[theorem]{Proposition}
\newtheorem{remark}[theorem]{Remark}

\makeatletter
\def\@seccntformat#1{\@ifundefined{#1@cntformat}%
	{\csname the#1\endcsname\quad}
	{\csname #1@cntformat\endcsname}
}
\makeatother

\newif\ifShowComments
\ShowCommentstrue
\def\strutdepth{\dp\strutbox}
\def\druk#1{\strut\vadjust{\kern-\strutdepth
        {\vtop to \strutdepth{%
                \baselineskip\strutdepth\vss
                        \llap{\hbox{#1}\quad}\null}}}}

\title{\bf
The $m$th Gini index estimator: Unbiasedness for gamma populations
}

\author{
\text{Roberto Vila}$^{1}$\thanks{Corresponding author: Roberto Vila, email: {rovig161@gmail.com}
}
\,\, and
\text{Helton Saulo}$^{1,2}$ 
\\
{\small $^{1}$ Department of Statistics, University of Brasilia, Brasilia, Brazil}\\
{\small $^{2}$ Department of Economics, Federal University of Pelotas, Pelotas, Brazil}\\
}
\begin{document}
	\maketitle 	
	\begin{abstract}
This paper establishes the theoretical result that the sample $m$th Gini index is an unbiased estimator of the population $m$th Gini index, introduced by \cite{Gavilan-Ruiz2024}, for gamma-distributed populations. An illustrative Monte Carlo simulation study confirms the unbiasedness of the sample $m$th Gini index estimator in finite samples.
	\end{abstract}
%

\section{Introduction}\label{sec:01}

Measurement of inequality plays a fundamental role in several areas of knowledge, with the classical Gini index being the most famous; see \cite{Gini1936}. This index, however, has an important limitation, as two distributions that behave differently in terms of their concentration can exhibit the same Gini index value. Recently, \cite{Gavilan-Ruiz2024} proposed a generalization of the classical Gini index to overcome this drawback by considering the difference between the expected maximum and minimum values of random samples of size $m$. This novel formulation preserves the intuitive attractiveness of the Gini index while allowing for a finer discrimination between distributions that might have identical Gini values but different internal structures of inequality.

In this paper, we prove the unbiasedness of the sample $m$th Gini index estimator when the underlying population follows a gamma distribution. The $m$th Gini index, recently introduced by \cite{Gavilan-Ruiz2024}, generalizes the classical Gini coefficient by considering the expected range between the maximum and minimum of $m$ independent observations. Our main result extends and generalizes the important finding of \cite{Baydil2025}, who established that the sample Gini coefficient is an unbiased estimator of the population Gini coefficient for a population having Gamma$(\alpha, \beta)$ distribution. Since the classical Gini coefficient corresponds to the case $m=2$ of the $m$th Gini index, our results not only recover the unbiasedness of the standard Gini coefficient as a particular case but also demonstrate that this property holds more broadly for any $m \geq 2$. The revelance of the results of this paper is to offer new theoretical insights into the behavior of inequality measures based on gamma populations and to strengthen the foundation for using $m$th Gini indices in practical applications.

The rest of this paper unfolds as follows. In Section \ref{sec:02}, we define the $m$th Gini index and derive its expression for gamma-distributed populations. In Section \ref{sec:03}, we establish the unbiasedness of the sample $m$th Gini index estimator and provide a closed-form expression for its expectation. In Section \ref{sec:04}, we present an illustrative Monte Carlo simulation study to corroborate our theoretical findings. Finally, in Section \ref{sec:05}, we provide concluding remarks.

\section{The $m$th Gini index for gamma distributions}\label{sec:02}
Let $X_1, X_2,\ldots, X_m$ be independent and identically distributed (iid) random variables with the same distribution as a non-negative random  $X$ with mean $\mu=\mathbb{E}(X)>0$. Throughout this work we will assume that $X_1, X_2,\ldots, X_m$ and $X$ are defined in the same sample space $\Omega$. Following Definition 1 of \cite{Gavilan-Ruiz2024}, for each integer $m\geqslant 2$, the $m$th Gini index of $X$ is defined as follows:
\begin{align}\label{m-pop-Gini}
	IG_m 
	\equiv IG_m(X)={\mathbb{E}\left[\max\{X_1,\ldots,X_m\}-\min\{X_1,\ldots,X_m\}\right]\over m\mu}.
\end{align}

\begin{remark}
Note that when $m=2$ the $m$th Gini index coincides with the classical definition of the Gini coefficient \citep{Gini1936}, denoted by $G$. That is,
\begin{align*}
G \equiv IG_2={\mathbb{E}[\, \vert X_1-X_2\vert \,]\over 2\mu}.
\end{align*}
\end{remark}

It is simple to observe that
\begin{align}\label{max-formula}
	\max\{X_{1},\ldots,X_{m}\}(\omega)
	&=
	\int_0^\infty 
	\mathds{1}_{[t,\infty)}( \max\{X_{1},\ldots,X_{m}\}(\omega))
	{\rm d}t
	\nonumber
	\\[0,2cm]
	&=
	\int_0^\infty 
	\left[
	1-
	\mathds{1}_{\bigcap_{k=1}^{m} X_{k}^{-1}((-\infty,t])}( \omega)
	\right]
	{\rm d}t,
	\quad 
	\forall\omega\in\Omega,
\end{align}
where $\mathds{1}_A$ is the indicator function of the set $A$ and $X_{k}^{-1}(B)=\{\omega\in\Omega: X_{k}(\omega)\in B \}$ is the inverse image of the Borel set $B$ through $X_{k}$. Similarly,
\begin{align}\label{min-formula}
	\min\{X_{1},\ldots,X_{m}\}(\omega)
	=
	\int_0^\infty 
	\mathds{1}_{\bigcap_{k=1}^{m} X_{k}^{-1}([t,\infty))}( \omega)
	{\rm d}t,
	\quad 
	\forall\omega\in\Omega.
\end{align}

Hence, by \eqref{max-formula},
\begin{align*}
	\mathbb{E}\left[\max\{X_1,\ldots,X_m\}\right]
	&=
	\int_0^\infty 
	\left[
	1-
	\mathbb{P}\left({\bigcap_{k=1}^{m} X_{k}^{-1}((-\infty,t])}\right)
	\right]
	{\rm d}t
	\\[0,2cm]
	&=
	\int_0^\infty 
	[
	1-
	\{\mathbb{P}\left(X^{-1}((-\infty,t])\right)\}^m
	]
	{\rm d}t
	=
	\int_0^\infty 
[
1-F^m_X(t)
]
{\rm d}t,
\end{align*}
where the fact that $X_1, X_2,\ldots, X_m$ are iid with the same distribution as $X$ was utilized. Analogously, by \eqref{min-formula},
\begin{align*}
	\mathbb{E}\left[\min\{X_1,\ldots,X_m\}\right]
	=
	\int_0^\infty [1-F_X(t)]^m {\rm d}t.
\end{align*}
Since 	$\mu=\mathbb{E}\left[X\right]
=
\int_0^\infty [1-F_X(t)] {\rm d}t$, the index $IG_m$ is characterized as follows:
\begin{align}\label{charact-m-Gini}
	IG_m
	=
	\dfrac{\displaystyle
		\int_0^\infty 
		[
		1-F^m_X(t)
		]
		{\rm d}t
		-
		\int_0^\infty [1-F_X(t)]^m {\rm d}t}{ \displaystyle
		m \int_0^\infty [1-F_X(t)] {\rm d}t}.
\end{align}

For reader convenience, in what follows we apply Equation \eqref{charact-m-Gini} to derive simple expressions for the $m$th index for the exponential model separately, despite it being a special case of gamma.
\begin{proposition}\label{nGini-exponential-prop}
	The $m$th Gini index for  $X\sim\exp(\lambda)$ (exponential distribution) is given by
	\begin{align}\label{nGini-exponential}
	IG_m 
	=
	{1\over m}
	\left[	
	{
		\sum_{k=1}^{m}
		\binom{m}{k}
		{(-1)^{k+1}\over k}-{1\over m}}
	\right].
\end{align}
\end{proposition}
\begin{proof}
	If $X\sim\exp(\lambda)$, then $\int_0^\infty [1-F_X(t)] {\rm d}t=\mathbb{E}[X]=1/\lambda$,
	$\int_0^\infty [1-F_X(t)]^m {\rm d}t=\int_0^\infty \exp(-m\lambda t) {\rm d}t=1/(m\lambda)$, and
	\begin{align*}
\int_0^\infty 
[
1-F^m_X(t)
]
{\rm d}t
&=
\int_0^\infty 
[
1-\{1-\exp(-\lambda t)\}^m
]
{\rm d}t
\\[0,2cm]
&=
\sum_{k=1}^{m}
\binom{m}{k}
(-1)^k
\int_0^\infty 
\exp(-k\lambda t)
{\rm d}t
=
{1\over \lambda}
\sum_{k=1}^{m}
\binom{m}{k}
{(-1)^{k+1}\over k},
	\end{align*}
	where Newton's binomial Theorem was applied.
Given the above identities and the characterization of $IG_m$ in \eqref{charact-m-Gini}, the result is immediate.
\end{proof}

\begin{remark}
	By taking $m=2$ in Proposition \ref{nGini-exponential-prop}, we get
$
G
		=
		IG_2
		=		
		{1/2}
$,
 for exponential distributions, consistent with prior research \citep{Deltas2003}.
\end{remark}

\begin{proposition}\label{nGini-gamma}
	The $m$th Gini index for  $X\sim\text{Gamma}(\alpha,\lambda)$ (gamma distribution) is given by
	\begin{align}\label{nGini-gamma-eq}
		IG_m 
		=
		{1\over m\alpha\Gamma^m(\alpha)} 
		\left[
		{	\displaystyle
			\int_0^\infty 
			\{
			\Gamma^m(\alpha)-\gamma^m(\alpha,s)
			\}
			{\rm d}s
			-
			\int_0^\infty 
			\Gamma^m(\alpha,s) {\rm d}s
			}
			\right],
	\end{align}
	where $\Gamma(a)$ is the (complete) gamma function, $\gamma(a,b)$ is the lower incomplete gamma function and $\Gamma(a,b)$ is the upper incomplete gamma function.
\end{proposition}
\begin{proof}
	If $X\sim\text{Gamma}(\alpha,\lambda)$, then $\int_0^\infty [1-F_X(t)] {\rm d}t=\mathbb{E}[X]=\alpha/\lambda$,
\begin{align*}
	&\int_0^\infty [1-F_X(t)]^m {\rm d}t
	=
	{1\over \Gamma^m(\alpha)}
	\int_0^\infty 
	\Gamma^m(\alpha,\lambda t) {\rm d}t,
	\\[0,2cm]
	&
	\int_0^\infty 
	[
	1-F^m_X(t)
	]
	{\rm d}t
	=
	{1\over \Gamma^m(\alpha)}
	\int_0^\infty 
	[
	\Gamma^m(\alpha)-\gamma^m(\alpha,\lambda t)
	]
	{\rm d}t.
\end{align*}	
Consequently, by using \eqref{charact-m-Gini} and the the change of variable $s = \lambda t$ yield the desired result.
\end{proof}

\begin{remark}
By applying the well-known identity $\gamma(a,b)+\Gamma(a,b)=\Gamma(a)$ and Newton's binomial Theorem, it follows that
\begin{align*}
	\Gamma^m(\alpha)-\gamma^m(\alpha,s)
	=
	\Gamma^m(\alpha)-[\Gamma(\alpha)-\Gamma(\alpha,s)]^m
	=
	\sum_{k=1}^{m}\binom{m}{k}(-1)^{k+1}\Gamma^{m-k}(\alpha) \Gamma^k(\alpha,s).
\end{align*}
Consequently,
\begin{align*}
				\int_0^\infty 
	[
	\Gamma^m(\alpha)-\gamma^m(\alpha,s)
	]
	{\rm d}s
	=
	\sum_{k=1}^{m}\binom{m}{k}(-1)^{k+1}\Gamma^{m-k}(\alpha)
				\int_0^\infty 
\Gamma^k(\alpha,s)
{\rm d}s.
\end{align*}

On the other hand, by using the definition
$\Gamma(a,x)=\int_x^\infty t^{a-1}\exp(-t){\rm d}t$ of the upper incomplete gamma function, we get
	\begin{align}\label{int-k}
				\int_0^\infty 
\Gamma^k(\alpha,s)
{\rm d}s
&=
				\int_0^\infty 
\int_{\mathbb{R}^k_+}
\mathds{1}_{\{s\leqslant\min\{t_1,\ldots,t_k\}\}}
(t_1\cdots t_k)^{\alpha-1}\exp\{-(t_1+\cdots+t_k)\}{\rm d}\boldsymbol{t}
{\rm d}s
\nonumber
\\[0,2cm]
&=
\int_{\mathbb{R}^k_+}
(t_1\cdots t_k)^{\alpha-1}
\exp\{-(t_1+\cdots+t_k)\}
\int_0^{\min\{t_1,\ldots,t_k\}} 
1
{\rm d}s 
{\rm d}\boldsymbol{t}
\nonumber
\\[0,2cm]
&=
\int_{\mathbb{R}^k_+}
\min\{t_1,\ldots,t_k\}
(t_1\cdots t_k)^{\alpha-1}
\exp\{-(t_1+\cdots+t_k)\}
{\rm d}\boldsymbol{t},
	\end{align}
	where Tonelli's Theorem validates the interchange of integrals.
	
Thus, $\int_0^\infty 
[
\Gamma^m(\alpha)-\gamma^m(\alpha,s)
]
{\rm d}s$ and $				\int_0^\infty 
\Gamma^m(\alpha,s)
{\rm d}s$ in Proposition \ref{nGini-gamma} can be written as
\begin{multline*}
				\int_0^\infty 
[
\Gamma^m(\alpha)-\gamma^m(\alpha,s)
]
{\rm d}s
\\[0,2cm]
=
\sum_{k=1}^{m}\binom{m}{k}(-1)^{k+1}\Gamma^{m-k}(\alpha)
\int_{\mathbb{R}^k_+}
\min\{t_1,\ldots,t_k\}
(t_1\cdots t_k)^{\alpha-1}
\exp\{-(t_1+\cdots+t_k)\}
{\rm d}\boldsymbol{t}
\end{multline*}
and
\begin{align*}
				\int_0^\infty 
\Gamma^m(\alpha,s)
{\rm d}s
=
\int_{\mathbb{R}^m_+}
\min\{t_1,\ldots,t_m\}
(t_1\cdots t_m)^{\alpha-1}
\exp\{-(t_1+\cdots+t_m)\}
{\rm d}\boldsymbol{t},
\end{align*}
respectively.
\end{remark}

\begin{remark}
	Taking $k=1$ in \eqref{int-k}, we have
	\begin{align*}
	\int_0^\infty 
	\Gamma(\alpha,s)
	{\rm d}s
	&=
	\int_{0}^\infty
	t_1^{\alpha}
	\exp(-t_1)
	{\rm d}t_1
	=
	\Gamma(\alpha+1)=\alpha\Gamma(\alpha),
	\end{align*}
	where in the last line the well-kwnon identity $\Gamma(x+1)=x\Gamma(x)$ was used.
	
    Setting $k=2$ in \eqref{int-k}, we get
	\begin{align}\label{id-int-gamma}
					\int_0^\infty 
	\Gamma^2(\alpha,s)
	{\rm d}s
	&=
	\int_{0}^\infty
	\int_{0}^\infty
	\min\{t_1,t_2\}
	(t_1t_2)^{\alpha-1}
	\exp\{-(t_1+t_2)\}
	{\rm d}t_1{\rm d}t_2
	\nonumber
	\\[0,2cm]
	&=
		\int_{0}^\infty \!
	t_2^{\alpha-1}
	\exp(-t_2)		
	\int_{0}^{t_2} \!
	t_1^{\alpha}
	\exp(-t_1)
	{\rm d}t_1{\rm d}t_2
	+
		\int_{0}^\infty \!
	t_2^{\alpha}
	\exp(-t_2)
	\int_{t_2}^\infty \!
	t_1^{\alpha-1}
	\exp(-t_1)
	{\rm d}t_1{\rm d}t_2
	\nonumber
			\\[0,2cm]
		&=
	\int_{0}^\infty
	t_2^{\alpha-1}
	\exp(-t_2)		
\gamma(\alpha+1,t_2)
	{\rm d}t_2
+
	\int_{0}^\infty
	t_2^{\alpha}
	\exp(-t_2)
[\Gamma(\alpha)-\gamma(\alpha,t_2)]
	{\rm d}t_2.
	\end{align}
		By using the identity \citep{DAurizio2016}:
		\begin{align*}
			\int_0^\infty 
			x^{a-1}
			\exp(-sx)\gamma(b,\theta x){\rm d}x
			=
			{\theta^b\Gamma(a+b)\over b(s+\theta)^{a+b}}
			\,_{2}F_{1}\left(a+b,1;b+1;{\theta\over s+\theta}\right),
		\end{align*}
		the identity in \eqref{id-int-gamma} becomes
	\begin{align}\label{id-int-gamma-1}
		\int_0^\infty 
		\Gamma^2(\alpha,s)
		{\rm d}s
		&=
	\alpha\Gamma^2(\alpha)
	+
	{\Gamma(2\alpha+1)\over (\alpha+1)2^{2\alpha+1}}
	\,_{2}F_{1}\left(2\alpha+1,1;\alpha+2;{1\over 2}\right)
	\nonumber
	\\[0,2cm]
	&-
	{\Gamma(2\alpha+1)\over \alpha 2^{2\alpha+1}}
	\,_{2}F_{1}\left(2\alpha+1,1;\alpha+1;{1\over 2}\right).
	\end{align}
	
By applying the identity \citep{WolframResearch2024}:
\begin{align*}
	\,_2F_1\left(a,b;{a+b\over 2}+1;{1\over 2}\right)
	=
	{2\sqrt{\pi}\Gamma({a+b\over 2}+1)\over a-b}
	\left[
	{1\over \Gamma({a\over 2})\Gamma({b+1\over 2})}
	-
	{1\over \Gamma({b\over 2})\Gamma({a+1\over 2})}
	\right],
\end{align*}
we obtain
\begin{align}\label{id-1}
	\,_{2}F_{1}\left(2\alpha+1,1;\alpha+2;{1\over 2}\right)
	=
		{\sqrt{\pi}\Gamma(\alpha+2)\over \alpha}
	\left[
	{1\over \Gamma({2\alpha+1\over 2})}
	-
	{1\over \Gamma({1\over 2})\Gamma(\alpha+1)}
	\right].
\end{align}
Furthermore, by using the identity \citep{WolframResearch2024}:
\begin{align*}
	\,_2F_1\left(a,b;{a+b\over 2};{1\over 2}\right)
	=
	\sqrt{\pi}\Gamma\left({a+b\over 2}\right)
	\left[
	{1\over \Gamma({a+1\over 2})\Gamma({b\over 2})}
+
{1\over \Gamma({a\over 2})\Gamma({b+1\over 2})}
	\right],
\end{align*}
we can write
\begin{align}\label{id-2}
	\,_{2}F_{1}\left(2\alpha+1,1;\alpha+1;{1\over 2}\right)
	=
	\sqrt{\pi}\Gamma\left(\alpha+1\right)
\left[
{1\over \Gamma(\alpha+1)\Gamma({1\over 2})}
+
{1\over \Gamma({2\alpha+1\over 2})}
\right].
\end{align}

Plugging \eqref{id-1} and \eqref{id-2} into \eqref{id-int-gamma-1} gives
	\begin{align}\label{iddd-1}
	\int_0^\infty 
	\Gamma^2(\alpha,s)
	{\rm d}s
	&=
	\alpha
	\Gamma^2(\alpha)
	-
	{\Gamma(2\alpha+1)\over 2^{2\alpha}\alpha}.
\end{align}

Thus, the integral $ \int_0^\infty [ \Gamma^2(\alpha)-\gamma^2(\alpha,s) ] {\rm d}s$ in Proposition \ref{nGini-gamma} becomes
\begin{align}\label{iddd-2}
	\int_0^\infty 
	[
	\Gamma^2(\alpha)-\gamma^2(\alpha,s)
	]
	{\rm d}s
	=
	2\Gamma(\alpha)
	\int_0^\infty 
	\Gamma(\alpha,s)
	{\rm d}s
	-
	\int_0^\infty 
	\Gamma^2(\alpha,s)
	{\rm d}s
	=
	\alpha
		\Gamma^2(\alpha)
	+
	{\Gamma(2\alpha+1)\over 2^{2\alpha}\alpha}.
\end{align}

Substituting \eqref{iddd-1} and \eqref{iddd-2} into Proposition \ref{nGini-gamma} yields
	\begin{align*}
	IG_2
	=
	{1\over 2\alpha\Gamma^2(\alpha)} 
	\left[
	{	\displaystyle
		\int_0^\infty 
		\{
		\Gamma^2(\alpha)-\gamma^2(\alpha,s)
		\}
		{\rm d}s
		-
		\int_0^\infty 
		\Gamma^2(\alpha,s) {\rm d}s
	}
	\right]
	=
		{\Gamma(2\alpha+1)\over 2^{2\alpha}\alpha^2\Gamma^2(\alpha)}.
\end{align*}
Finally, by applying Legendre duplication formula  \citep{Abramowitz1972}:
$
\Gamma(x)\Gamma\left(x+{1/ 2}\right)=2^{1-2x}\sqrt{\pi}\Gamma(2x),
$ in the above identity,
we get the following expression for the $2$th Gini index:
\begin{align}\label{main-formula}
	G
	=
	IG_2
	=
		{1\over 2\alpha\Gamma^2(\alpha)} 
	\left[
	{	\displaystyle
		\int_0^\infty 
		\{
		\Gamma^2(\alpha)-\gamma^2(\alpha,s)
		\}
		{\rm d}s
		-
		\int_0^\infty 
		\Gamma^2(\alpha,s) {\rm d}s
	}
	\right]
	=
	{\Gamma(\alpha+{1\over 2})\over\sqrt{\pi}\alpha\Gamma(\alpha)},
\end{align}
which corroborates the Gini coefficient formula for gamma distributions reported in \cite{McDonald1979}.
\end{remark}

\section{Unbiasedness of the $m$th Gini index estimator}\label{sec:03}

In this section (more precisely, in Theorem \ref{main-theorem}), we provide a general formula 
 for the expectation of the $m$th Gini index estimator,  newly defined in this work, as follows (for $m\leqslant n$)
\begin{align}\label{estimator}
\widehat{IG}_m
=
{(m-1)!\over (n-1)(n-2)\cdots(n-m+1)} \,
\dfrac{\displaystyle
\sum_{1\leqslant i_1<\cdots< i_m\leqslant n}
\left[
\max\{X_{i_1},\ldots,X_{i_m}\}-\min\{X_{i_1},\ldots,X_{i_m}\}
\right]
}{\displaystyle \sum_{i=1}^{n}X_i},
\end{align}
where $X_1, X_2,\ldots, X_m$ are iid observations from the  population $X$. 
\begin{remark}
Setting $m=2$ in \eqref{estimator}, we get the estimator of the Gini coefficient, denoted by $\widehat{G}$,
\begin{align*}
	\widehat{G}
	\equiv
\widehat{IG}_2
=
{1\over n-1} \,
\dfrac{\displaystyle\sum_{1\leqslant i<j\leqslant n} \vert X_i-X_j\vert}{\displaystyle \sum_{i=1}^{n}X_i},
\end{align*}
which initially was proposed by \cite{Deltas2003}.
\end{remark}

\begin{theorem}\label{main-theorem}
	Let $X_1, X_2,\ldots,X_m$ be independent copies of a  non-negative and absolutely continuous random variable $X$ with finite and positive expected value and common cumulative distribution function $F$. The following holds:
\begin{align*}
\mathbb{E}[\widehat{IG}_m]
&=
{n\over m}
\int_0^\infty
\int_0^\infty 
\left\{
\mathscr{L}_F^{m}(z)
-
\mathbb{E}^m\left[
\mathds{1}_{X^{-1}((-\infty,t])}
\exp\left(-Xz\right)
\right]
\right\}
{\rm d}t \,
\mathscr{L}_F^{n-m}(z)
{\rm d}z
\\[0,2cm]
&-
{n\over m}
\int_0^\infty
\int_0^\infty 
\mathbb{E}^m\left[
\mathds{1}_{X^{-1}([t,\infty))}
\exp\left(-Xz\right)
\right]
{\rm d}t \,
\mathscr{L}_F^{n-m}(z)
{\rm d}z,
\end{align*}
	where $\mathscr{L}_F(z)=\int_0^\infty \exp(-zx){\rm d}F(x)$ is the Laplace transform corresponding to distribution $F$.  In the
	above, we are assuming that the expectations and improper integrals involved exist.
\end{theorem}
\begin{proof}
By using the identity
\begin{align*}
\int_{0}^\infty \exp(-\xi z){\rm d}z={1\over \xi},
\quad \xi>0,
\end{align*}
with $\xi=\sum_{i=1}^{n}X_i$, we get
\begin{multline*}
	\mathbb{E}\left[
	\dfrac{\displaystyle
		\sum_{1\leqslant i_1<\cdots< i_n\leqslant n}
		\left[
		\max\{X_{i_1},\ldots,X_{i_m}\}-\min\{X_{i_1},\ldots,X_{i_m}\}
		\right]
	}{\displaystyle \sum_{i=1}^{n}X_i}
	\right]
	\\[0,2cm]
	=
	\sum_{1\leqslant i_1<\cdots< i_m\leqslant n}
	\mathbb{E}\left[\max\{X_{i_1},\ldots,X_{i_m}\}\int_0^\infty\exp\left\{-\left(\sum_{i=1}^{n}X_i\right)z\right\}{\rm d}z\right]
		\\[0,2cm]
-
	\sum_{1\leqslant i_1<\cdots< i_m\leqslant n}
	\mathbb{E}\left[\min\{X_{i_1},\ldots,X_{i_m}\}\int_0^\infty\exp\left\{-\left(\sum_{i=1}^{n}X_i\right)z\right\}{\rm d}z\right].
\end{multline*}
By applying identities in \eqref{max-formula} and \eqref{min-formula}, the above expression becomes
\begin{align*}
&=\sum_{1\leqslant i_1<\cdots< i_m\leqslant n}
\mathbb{E}\left[	
\int_0^\infty 
\left\{
1-
\mathds{1}_{\bigcap_{k=1}^{m} X_{i_k}^{-1}((-\infty,t])}
\right\}
{\rm d}t\int_0^\infty\exp\left\{-\left(\sum_{i=1}^{n}X_i\right)z\right\}{\rm d}z\right]
\\[0,2cm]
&-
\sum_{1\leqslant i_1<\cdots< i_m\leqslant n}
\mathbb{E}\left[
	\int_0^\infty 
\mathds{1}_{\bigcap_{k=1}^{m} X_{i_k}^{-1}([t,\infty))}
{\rm d}t
\int_0^\infty\exp\left\{-\left(\sum_{i=1}^{n}X_i\right)z\right\}{\rm d}z\right]
\\[0,2cm]
&=
	\sum_{1\leqslant i_1<\cdots< i_m\leqslant n} \!
	\int_0^\infty \!
\int_0^\infty \!\!
\mathbb{E}\left[	\!
\left\{
1-
\mathds{1}_{\bigcap_{k=1}^{m} X_{i_k}^{-1}((-\infty,t])}
\right\} \!
\exp\left\{\!-\left(\sum_{i=1}^{m}X_i\right)z\right\}
\exp\left\{\!-\left(\sum_{i=m+1}^{n}X_i\right)\! z\right\}
\right] \!
{\rm d}t
{\rm d}z
\\[0,2cm]
&-
\sum_{1\leqslant i_1<\cdots< i_m\leqslant n}
\int_0^\infty
\int_0^\infty 
\mathbb{E}\left[
\mathds{1}_{\bigcap_{k=1}^{m} X_{i_k}^{-1}([t,\infty))}
\exp\left\{-\left(\sum_{i=1}^{m}X_i\right)z\right\}
\exp\left\{-\left(\sum_{i=m+1}^{n}X_i\right)z\right\}
\right]
{\rm d}t
{\rm d}z,
\end{align*}
where the interchange of integrals is justified by Tonelli's Theorem. Since $X_1, X_2,\ldots, X_m$  are iid, the last expression simplifies to
\begin{align*}
	&=
\binom{n}{m} \!
	\int_0^\infty \! \!
	\int_0^\infty \!
	\mathbb{E}\left[	 \!
	\left\{
	1-
	\mathds{1}_{\bigcap_{k=1}^{m} X_{k}^{-1}((-\infty,t])}
	\right\} \!
	\exp\left\{-\left(\sum_{i=1}^{m}X_i\right)z\right\}
	\mathbb{E}\left[
	\exp\left\{-\left(\sum_{i=m+1}^{n}X_i\right)z\right\}
	\right]
	\right]
	{\rm d}t
	{\rm d}z
	\\[0,2cm]
	&-
\binom{n}{m}
	\int_0^\infty
	\int_0^\infty 
	\mathbb{E}\left[
	\mathds{1}_{\bigcap_{k=1}^{m} X_{k}^{-1}([t,\infty))}
	\exp\left\{-\left(\sum_{i=1}^{m}X_i\right)z\right\}
	\mathbb{E}\left[
	\exp\left\{-\left(\sum_{i=m+1}^{n}X_i\right)z\right\}
	\right]
	\right]
	{\rm d}t
	{\rm d}z
	\\[0,2cm]
	&=
	\binom{n}{m}
	\int_0^\infty
	\int_0^\infty 
	\mathbb{E}\left[	
	\left\{
	1-
	\mathds{1}_{\bigcap_{k=1}^{m} X_{k}^{-1}((-\infty,t])}
	\right\}
	\exp\left\{-\left(\sum_{i=1}^{m}X_i\right)z\right\}
	\right]
	{\rm d}t \,
	\mathscr{L}_F^{n-m}(z)
	{\rm d}z
	\\[0,2cm]
	&-
	\binom{n}{m}
	\int_0^\infty
	\int_0^\infty 
	\mathbb{E}\left[
	\mathds{1}_{\bigcap_{k=1}^{m} X_{k}^{-1}([t,\infty))}
	\exp\left\{-\left(\sum_{i=1}^{m}X_i\right)z\right\}
	\right]
	{\rm d}t \,
	\mathscr{L}_F^{n-m}(z)
	{\rm d}z
	\\[0,2cm]
	&=
	\binom{n}{m}
\int_0^\infty
\int_0^\infty 
\left\{
\mathscr{L}_F^{m}(z)
-
\mathbb{E}^m\left[
\mathds{1}_{X^{-1}((-\infty,t])}
\exp\left(-Xz\right)
\right]
\right\}
{\rm d}t \,
\mathscr{L}_F^{n-m}(z)
{\rm d}z
\\[0,2cm]
&-
	\binom{n}{m}
\int_0^\infty
\int_0^\infty 
\mathbb{E}^m\left[
\mathds{1}_{X^{-1}([t,\infty))}
\exp\left(-Xz\right)
\right]
{\rm d}t \,
\mathscr{L}_F^{n-m}(z)
{\rm d}z,
	\end{align*}
where, again, the iid assumption for $X_1,X_2,\ldots,X_m$ was used. Finally, by using the definition of the $m$th Gini estimator, given in \eqref{estimator}, the proof of theorem follows.
\end{proof}

In the remainder of this section we apply Theorem \ref{main-theorem} to derive explicit expressions for the expectation of estimator $\widehat{IG}_m$ in exponential and gamma populations, demonstrating the estimator's unbiasedness in both cases. Although the exponential case is a special case of gamma, we present both for clarity.

\begin{corollary}\label{main-theorem-1}
	Let $X_1, X_2,\ldots,X_m$ be independent copies of $X\sim\exp(\lambda)$. We have:
\begin{align*}
		\mathbb{E}[\widehat{IG}_m]
		=		
		{1\over m} \, 
		\left[
		\sum_{k=1}^{m}
		\binom{m}{k}
		{(-1)^{k+1} \over k}
		-
		{1\over m}
		\right]
		=
		IG_m,
\end{align*}
where $IG_m$ is the $m$th Gini index given in \eqref{nGini-exponential}. Hence, the estimator $\widehat{IG}_m$ is unbiased under exponential populations.
\end{corollary}
\begin{proof}
Assuming  $X\sim\exp(\lambda)$, simple calculation yields
	\begin{align*}
\mathbb{E}\left[
\mathds{1}_{X^{-1}((-\infty,t])}
\exp\left(-Xz\right)
\right]
=
{\lambda\over z+\lambda} \,
[1-\exp\left\{-(z+\lambda)t\right\}] 
	\end{align*}
	and
	\begin{align*}
	\mathbb{E}\left[
	\mathds{1}_{X^{-1}([t,\infty))}
	\exp\left(-Xz\right)
	\right]
	=
{\lambda\over z+\lambda}
\, 
\exp\left\{-(z+\lambda)t\right\}.
	\end{align*}
	As $\mathscr{L}_F(z)=\lambda/(z+\lambda)$, applying Theorem \ref{main-theorem} we have
\begin{align}\label{ngini-exp}
	\mathbb{E}[\widehat{IG}_m]
	&=
	{n\lambda^n\over m}
	\int_0^\infty
	\int_0^\infty 
	\left\{
	1
	-
	\left[
	1-\exp\left\{-(z+\lambda)t\right\}
	\right]^m
	\right\}
	{\rm d}t \,
	{1\over (z+\lambda)^n}
	{\rm d}z
	\nonumber
	\\[0,2cm]
	&-
	{n\lambda^n\over m}
	\int_0^\infty
	\int_0^\infty 
	\exp\left\{-m(z+\lambda)t\right\}
	{\rm d}t \,
	{1\over (z+\lambda)^n}
	{\rm d}z.
\end{align}

Using Newton's binomial Theorem, we obtain
\begin{align*}
		\int_0^\infty 
	\left\{
	1
	-
	\left[
	1-\exp\left\{-(z+\lambda)t\right\}
	\right]^m
	\right\}
	{\rm d}t
	&=
		\sum_{k=1}^{m}
	\binom{m}{k}
	(-1)^{k+1} 
	\int_0^\infty 
	\exp\left\{-k(z+\lambda)t\right\}
	{\rm d}t
	\\[0,2cm]
	&=
	{1\over z+\lambda}
		\sum_{k=1}^{m}
	\binom{m}{k}
	{(-1)^{k+1} \over k}.
\end{align*}
Plugging the above identity into \eqref{ngini-exp} gives
\begin{align}
	\mathbb{E}[\widehat{IG}_m]
	&=
	{n\lambda^n\over m} \, 
		\sum_{k=1}^{m}
\binom{m}{k}
{(-1)^{k+1} \over k}
	\int_0^\infty
	{1\over (z+\lambda)^{n+1}}
	{\rm d}z
	-
	{n\lambda^n\over m} \,
	{1\over m}
	\int_0^\infty
	{1\over (z+\lambda)^{n+1}}
	{\rm d}z
	\nonumber
	\\[0,2cm]
	&=
		{1\over m} \, 
		\left[
	\sum_{k=1}^{m}
	\binom{m}{k}
	{(-1)^{k+1} \over k}
	-
	{1\over m}
	\right]
	=	
	{IG}_m,
\end{align}
where the last equality follows from Equation \eqref{nGini-exponential}.
This completes the proof.
\end{proof}

\begin{remark}
	By setting $m = 2$ in Corollary \ref{main-theorem-1} yields
\begin{align*}
	\mathbb{E}[\widehat{G}]
	=		
	{1\over 2} 
	=
	G,
\end{align*}	
 for exponential distributions,
	which is a well-established result in the literature \citep{Deltas2003}.
\end{remark}

\begin{corollary}\label{main-theorem-2}
	Let $X_1, X_2,\ldots,X_m$ be independent copies of $X\sim\text{Gamma}(\alpha,\lambda)$. We have:
	\begin{align*}
		\mathbb{E}[\widehat{IG}_m]
		=		
		{1\over m\alpha\Gamma^m(\alpha)} 
\left[
{	\displaystyle
	\int_0^\infty 
	\{
	\Gamma^m(\alpha)-\gamma^m(\alpha,s)
	\}
	{\rm d}s
	-
	\int_0^\infty 
	\Gamma^m(\alpha,s) {\rm d}s
}
\right]
		=
		IG_m,
	\end{align*}
	where $IG_m$ is the $m$th Gini index given in Proposition \ref{nGini-gamma}. Therefore, the estimator $\widehat{IG}_m$ is unbiased for  gamma populations.
\end{corollary}
\begin{proof}
	When $X\sim\text{Gamma}(\alpha,\lambda)$, direct computation shows
	\begin{align*}
		\mathbb{E}\left[
		\mathds{1}_{X^{-1}((-\infty,t])}
		\exp\left(-Xz\right)
		\right]
		=
		{\lambda^\alpha\over (z+\lambda)^\alpha \Gamma(\alpha)} \,
		\gamma(\alpha,(z+\lambda)t) 
	\end{align*}
	and
	\begin{align*}
		\mathbb{E}\left[
		\mathds{1}_{X^{-1}([t,\infty))}
		\exp\left(-Xz\right)
		\right]
		=
		{\lambda^\alpha\over (z+\lambda)^\alpha \Gamma(\alpha)} 
		\, 
		\Gamma(\alpha,(z+\lambda)t).
	\end{align*}
	Given $\mathscr{L}_F(z) = \lambda^\alpha/(z + \lambda)^\alpha$, Theorem \ref{main-theorem} yields
	\begin{align*}
		\mathbb{E}[\widehat{IG}_m]
		&=
		{n\over m} \, 
		{\lambda^{\alpha n} \over \Gamma^m(\alpha)}
		\int_0^\infty
		\int_0^\infty 
		\left\{
		\Gamma^m(\alpha)
		-
\gamma^m(\alpha,(z+\lambda)t) 
		\right\}
		{\rm d}t \,
		{1\over (z+\lambda)^{\alpha n}}
		{\rm d}z
		\\[0,2cm]
		&-
		{n\over m} \,
		{\lambda^{\alpha n}\over \Gamma^m(\alpha)}
		\int_0^\infty
		\int_0^\infty 
\Gamma^m(\alpha,(z+\lambda)t) 
		{\rm d}t \,
		{1\over (z+\lambda)^{\alpha n}}
		{\rm d}z.
	\end{align*}
	Now, making the change of variable $u=(z+\lambda)t$, the above identity can be written as
		\begin{align*}
		\mathbb{E}[\widehat{IG}_m]
		&=
		{n\over m} \, 
		{\lambda^{\alpha n} \over \Gamma^m(\alpha)}
		\int_0^\infty
		\int_0^\infty 
		\left\{
		\Gamma^m(\alpha)
		-
		\gamma^m(\alpha,u) 
		\right\}
		{\rm d}u \,
		{1\over (z+\lambda)^{\alpha n+1}}
		{\rm d}z
		\\[0,2cm]
		&-
		{n\over m} \,
		{\lambda^{\alpha n}\over \Gamma^m(\alpha)}
		\int_0^\infty
		\int_0^\infty 
		\Gamma^m(\alpha,u) 
		{\rm d}u \,
		{1\over (z+\lambda)^{\alpha n+1}}
		{\rm d}z
		\\[0,2cm]
		&=
				{1\over m\alpha\Gamma^m(\alpha)} 
		\left[
		{	\displaystyle
			\int_0^\infty 
			\{
			\Gamma^m(\alpha)-\gamma^m(\alpha,u)
			\}
			{\rm d}u
			-
			\int_0^\infty 
			\Gamma^m(\alpha,u) {\rm d}u
		}
		\right]
		=
		IG_m,
	\end{align*}
	where the last equality follows from Proposition \ref{nGini-gamma}.
	Thus, the proof is complete.
\end{proof}

\begin{remark}	
	By taking $m = 2$ in Corollary \ref{main-theorem-2}, it follows from \eqref{main-formula} that
	\begin{align*}
		\mathbb{E}[\widehat{G}]
		=
				\mathbb{E}[\widehat{IG}_2]
		=		
		{1\over 2\alpha\Gamma^2(\alpha)} 
\left[
{	\displaystyle
	\int_0^\infty 
	\{
	\Gamma^2(\alpha)-\gamma^2(\alpha,s)
	\}
	{\rm d}s
	-
	\int_0^\infty 
	\Gamma^2(\alpha,s) {\rm d}s
}
\right]
		=
			{\Gamma(\alpha+{1\over 2})\over\sqrt{\pi}\alpha\Gamma(\alpha)}
			=
		G,
	\end{align*}	
	 for gamma distributions,
	which is a well-known result in the literature  \citep[see, e.g,][]{Baydil2025,Vila2025}.
\end{remark}

\section{Illustrative simulation study}\label{sec:04}

In this section, we present an illustrative Monte Carlo simulation study to evaluate the performance of the $m$th Gini index estimator $\widehat{IG}_m$ in finite samples. Specifically, we want to confirm via simulations the unbiasedness of the $m$th Gini index estimator. We assess the bias and mean squared error (MSE) of the estimator for $m=3$ under two settings: exponential and gamma populations. The bias and MSE are computed as follows:
\[
\widehat{\text{Bias}}(\widehat{IG}_m) = \frac{1}{N_{\text{sim}}} \sum_{k=1}^{N_{\text{sim}}} {\widehat{IG}_m^{(k)} - IG_m},
\quad
\widehat{\text{MSE}}(\widehat{IG}_m) = \frac{1}{N_{\text{sim}}} \sum_{k=1}^{N_{\text{sim}}} \big(\widehat{IG}_m^{(k)} - IG_m\big)^2,
\]
where $\widehat{IG}_m^{(k)}$ denotes the $k$-th Monte Carlo replicate of the sample $m$th Gini index estimator defined in Equation~\eqref{estimator}, and $IG_m$ represents the true $m$th Gini index defined in Equation~\eqref{nGini-exponential} for the exponential distribution and in Equation~\eqref{nGini-gamma-eq} for the gamma distribution. Here, $N_{\text{sim}}$ is the number of Monte Carlo replications.

For the exponential case, random samples were generated from the $\text{Exp}(\lambda=1)$ distribution, whereas for the gamma case, samples were drawn from the $\text{Gamma}(\alpha=2, \lambda=1)$ distribution. For each considered sample size $n \in \{5, 10, 30, 50, 100\}$, we conducted $N_{\text{sim}}=1000$ replications. The parameter $\alpha$ in Equation~\eqref{nGini-gamma-eq} is assumed to be known.

Table~\ref{tab:mc-results} presents the results of the Monte Carlo simulation study. Consistent with the theoretical unbiasedness established in Section~\ref{sec:03}, the empirical bias is close to zero and decreases as the sample size increases. Similarly, the MSE decreases as the sample size grows.


\begin{table}[!h]
\centering
\caption{Empirical bias and MSE of $\widehat{IG}_m$ ($m=3$) based on the indicated distributions.}
\label{tab:mc-results}
\begin{tabular}[t]{cccc}
\toprule
Distribution & $n$ & Bias & MSE \\
\midrule
\rowcolor{gray!10} Exponential & 5 & -0.00395 & 0.02131 \\
Exponential & 10 & -0.00179 & 0.00929 \\
\rowcolor{gray!10} Exponential & 30 & -0.00133 & 0.00306 \\
Exponential & 50 & -0.00226 & 0.00174 \\
\rowcolor{gray!10} Exponential & 100 & 0.00004 & 0.00087 \\
\addlinespace
\rowcolor{gray!10} Gamma & 5 & 0.00866 & 0.01604 \\
Gamma & 10 & 0.00441 & 0.00664 \\
\rowcolor{gray!10} Gamma & 30 & -0.00073 & 0.00191 \\
Gamma & 50 & 0.00236 & 0.00111 \\
\rowcolor{gray!10} Gamma & 100 & 0.00005 & 0.00053 \\
\bottomrule
\end{tabular}
\end{table}

\section{Concluding remarks}\label{sec:05}

We have studied the theoretical properties of the $m$th Gini index and its sample estimator $\widehat{IG}_m$, which generalizes the classical Gini index by considering differences between the maximum and minimum values of random samples of size $m$. We have shown that $\widehat{IG}_m$ is an unbiased estimator of the corresponding population parameter $IG_m$ under gamma-distributed populations. Our results have extended and generalized the findings of \cite{Baydil2025}, who investigated the unbiasedness of the traditional Gini index estimator for the population Gini coefficient in gamma distributions. We have validated our theoretical results through an illustrative Monte Carlo simulation study. The simulation confirmed the unbiasedness of the estimator $\widehat{IG}_m$ in finite samples, under both exponential and gamma settings.

%


\paragraph*{Acknowledgements}
The research was supported in part by CNPq and CAPES grants from the Brazilian government.

\paragraph*{Disclosure statement}
There are no conflicts of interest to disclose.





\begin{thebibliography}{}
	
	\bibitem[Abramowitz and Stegun, 1972]{Abramowitz1972}
	Abramowitz, M. and Stegun, I.~A. (1972).
	\newblock {\em Handbook of Mathematical Functions with Formulas, Graphs, and
		Mathematical Tables}.
	\newblock Dover, New York, 9th printing edition.
	
	\bibitem[Baydil et~al., 2025]{Baydil2025}
	Baydil, B., de~la Peña, V.~H., Zou, H., and Yao, H. (2025).
	\newblock Unbiased estimation of the gini coefficient.
	\newblock {\em Statistics and Probability Letters}.
	
	\bibitem[D'Aurizio, 2016]{DAurizio2016}
	D'Aurizio, J. (2016).
	\newblock Integrating the lower incomplete gamma $\int_0^\infty x^{a-1}e^{-s x}
	\gamma(b,x) \mathrm{d}x$.
	\newblock {\em Mathematics Stack Exchange}.
	
\bibitem[Deltas, 2003]{Deltas2003}
Deltas, G. (2003).
\newblock The small-sample bias of the gini coefficient: Results and
implications for empirical research.
\newblock {\em Review of Economics and Statistics}, 85:226--234.

\bibitem[Gavilan-Ruiz et al., 2024]{Gavilan-Ruiz2024}
Gavilan-Ruiz, J. M.,
Ruiz-Gándara, Á., Ortega-Irizo, F. J.,
Gonzalez-Abril, L. (2024).
Some Notes on
the Gini Index and New Inequality
Measures: The nth Gini Index. 
Stats, 7:1354--1365. 
\url{https://doi.org/10.3390/stats7040078}.

\bibitem[Gini, 1936]{Gini1936}
Gini, C. (1936).
\newblock On the measure of concentration with special reference to income and
statistics.
\newblock {\em Colorado College Publication, General Series No. 208}, pages
73--79.

\bibitem[McDonald and Jensen, 1979]{McDonald1979}
McDonald, J.~B. and Jensen, B.~C. (1979).
\newblock An analysis of some properties of alternative measures of income
inequality based on the gamma distribution function.
\newblock {\em Journal of the American Statistical Association}, 74:856--860.

\bibitem[Vila and Saulo, 2025]{Vila2025}
Vila, R. and Saulo, H. (2025). Bias in gini coefficient estimation for gamma mixture populations.
\url{https://arxiv.org/pdf/2503.00690}. Manuscript submitted for publication. Preprint.

\bibitem[Wolfram~Research, 2024]{WolframResearch2024}
Wolfram~Research, I. (2024).
\newblock {\em Mathematica, Version 14.2}.
\newblock Champaign, IL.



%
%
%
%
%
%
%
%
%

\end{thebibliography}

\end{document}